\documentclass[12pt]{amsart}
\usepackage{graphicx}
\usepackage{color}
\newtheorem{lemma}{Lemma}

\begin{document}

\title{Generalization of nonlinear control for nonlinear discrete systems}
\author{D. Dmitrishin, A. Stokolos, I. Skrynnik, E. Franzheva}

\begin{abstract}  
The problem of stabilization of unstable periodic orbits of discrete nonlinear systems is considered in the article. A new 
generalization of the delayed feedback, which solves the stabilization problem, is proposed. The feedback is represented as 
a convex combination of nonlinear feedback and semilinear feedback introduced by O. Morgul. In this article, the O. Morgul 
method was transferred from the scalar case to the vector one. It is shown that the additional introduction of the semilinear feedback into the equation makes it possible to significantly reduce the length of the prehistory used in the 
control and to increase the rate of convergence of the perturbed solutions to periodic ones. As an application of the 
proposed stabilization scheme, a possible computational algorithm for finding solutions of systems of algebraic equations 
is given. The numerical simulation results are presented.
\end{abstract}

\maketitle

\section{Introduction}

By control of chaos we mean a small external influence on the system or a small change in the structure of the system in order to transform the chaotic behavior of the system into regular (or chaotic, but with different properties) [1]. The problem of optimal influence on the chaotic regime is one of the fundamental problems in nonlinear dynamics [2, 3].

It is assumed that the dynamic system has a chaotic attractor, which contains a countable set of unstable cycles of different periods. If the control action locally stabilizes a cycle, then the trajectory of the system remains in its neighborhood, i.e. regular movements will be observed in the system. Hence, one of the ways to control chaos is the local stabilization of certain orbits from a chaotic attractor.

To solve the stabilization problem, various control schemes were proposed [4], among which controls based on the Delayed Feedback Control (DFC) principle are quite popular [5]. Such controls, under certain conditions, allow local stabilization of equilibrium positions or cycles, which, generally speaking, are not known in advance. Among the DFC schemes, linear schemes are the simplest for physical implementation. However, they have significant limitations: they can be used only for a narrow area of the parameter space that enter the original nonlinear system. The necessary conditions for the applicability of the linear feedback are formulated more precisely in Section 2.1.

To extend the class of systems to which the DFC scheme applies, it is necessary to introduce non-linear elements into the control. For the first time, a nonlinear DFC with one delay was considered in [6], where the advantages of such a modification are also noted, in particular, the fact that the control becomes robust. In [7, 8] the concept of nonlinear control with one delay from [6] was extended: to the vector case; to manage with several delays; to the case of an arbitrary period $T$. It is shown that the control allows to stabilize cycles of arbitrary lengths, unless the multipliers are real and greater than one. A relationship is established between the size of the localization set of multipliers and the amount of delay in the nonlinear feedback.

In [9, 10], a semilinear DFC scheme with linear and nonlinear elements was investigated. In spite of the fact that this scheme contains only one difference in control, nevertheless, it is possible to stabilize cycles with length $T = 1, 2$ under sufficiently general assumptions about cycle multipliers. For $T\ge 3$, the situation changes critically, and the stabilization of cycles is possible only if the rigid constraints on multipliers are met. The scheme of O. Morgul is considered in detail in Section 2.3. It will be generalized to the case of several differences in control, and transferred from scalar to vector case.

The purpose of the presented work is to improve the algorithms of M.Vieira de Souza, A.J. Lichtenberg, O. Morgul, and D. Dmitrishin: suppression of chaos in nonlinear discrete systems by local stabilization of cycles of a given length.

Accordingly, the problem consists of choosing the structure and parameters of the control system, in which beforehand unknown cycles of a given length would be locally asymptotically stable.

The paper considers delayed feedback in the form of a convex combination of nonlinear control and generalized control by O. Morgul. The characteristic polynomial of a closed system for a cycle of length $T$ is deduced and its structure has turned out to be quite simple. As a particular case, this polynomial contains the characteristic polynomials for non-linear control and generalized control of O. Morgul. The solution of the problem for stabilizing cycles of length one is given, i.e. equilibrium positions, and a theoretical basis is prepared for solving the problem in a general formulation for cycles of arbitrary length.

The special structure of the characteristic polynomial allows the use of methods of complex analysis. That is why the main method of constructing controls and investigating the conditions for their applicability is the geometric theory of functions of a complex variable. From the perspective of this theory, O. Morgul's approach to stabilizing cycles and the conditions for its applicability are analyzed. The analysis of the influence of the control parameters on the quality of control is carried out and it is indicated why the combined control is better than the nonlinear or semi-linear control separately. Finally, applications of the proposed scheme of combined control to the improvement of iterative methods for solving algebraic equations are considered.

\section{Review and preliminary results}

We consider a nonlinear discrete system, which in the absence of control has the form
\begin{equation}\label{(1)}
{x_{n + 1}} = f\left({x_n}\right), \quad {x_n} \in {\mathbb R^m}, \quad n = 1, 2, \ldots,
\end{equation}
where $f\left(x\right)$ is a differentiable vector function of the corresponding dimension. It is assumed that the system (1) has an invariant convex set $A$, that is, if $\xi \in A$, then $f\left(\xi\right) \in A$. It is also assumed that in this system there is one or more unstable $T$-cycles $\left({{\eta_1}, \ldots, {\eta_T}}\right)$, where all the vectors ${\eta_1}, \ldots, {\eta_T}$ are distinct and belong to the invariant set $A$, i.e. ${\eta_{j + 1}} = f\left({{\eta_j}}\right), j = 1, \ldots, T - 1, {\eta_1} = f\left({{\eta_T}}\right)$. The multipliers of the unstable cycles under consideration are defined as the eigenvalues of the products of the Jacobian matrices $\prod_{j = 1}^T {{{f'}_{}}\left({{\eta_j}}\right)}$ of dimensions $m \times m$. As a rule, the cycles $\left({{\eta_1}, \ldots, {\eta_T}}\right)$ of the system (1) are not a priori known as well as the spectrum of the matrix $\left\{ {{\mu_1}, \ldots, {\mu_m} } \right\}$ of the matrix $\prod_{j = 1}^T {f'\left({{\eta_j}}\right)}$.

It is required to describe the set $M$ in which it is possible to locally stabilize the $T$-cycle of the system (1) by one control from the admissible control class for all multipliers localized in $M$, $M \subset \bar C$ ($\bar C$ is the extended complex plane). I.e., so that the system
\[
{x_{n + 1}} = f\left({x_n}\right) + {u_n}
\]
would have a locally asymptotically stable $T$-cycle with multipliers in $M$, and on this cycle the control ${u_n}$ would vanish. In other words, we assume that for a given cycle length $T$, we know the estimate of the localization set of the multipliers $M$. In other words, we believe that the dynamic system is characterized not so much by the function $f$ (or a family of functions) as by the set of localization of multipliers of a cycle (or cycles) of known length.

\subsection{Linear control}

As a control, let us consider a law based on linear feedback
\begin{equation}\label{(2)}
{u_n} = - \sum_{j = 1}^{N - 1} {{\varepsilon_j}\left({{x_{n-jT+T}} - {x_{n-jT}}}\right)}, 
\end{equation}
where the gain should be limited: $\left|{{\varepsilon_j}} \right| < 1$, $j = 1, \ldots, N - 1$, $T = 1, 2, \ldots$. Accordingly, the system closed by such a control has the form
\begin{equation}\label{(3)}
{x_{n + 1}} = f\left({x_n}\right) - \sum_{j = 1}^{N - 1} {{\varepsilon_j}\left({{x_{n - jT + T}} - {x_{n - jT}}}\right)}.
\end{equation}
Note that when the state ${x_{k + T}} = {x_k}$, $k = 1, 2, \ldots$ is synchronized, the control (\ref{(2)}) vanishes, i.e. the closed system (3) takes the form as in the absence of control. This means that the $T$-cycles of the system (1) are $T$-cycles of the system (3).

Consider the case $T = 1$. It is required to find the necessary conditions in terms of the localization set of the multipliers $M$ for which the equilibrium position of the system (3) is locally asymptotically stable (or sufficient conditions under which this equilibrium position is unstable). It is shown in [11] that the set $M$ of localization of multipliers of system (1) can not be arbitrarily large for any linear control of the form (2), more precisely, its diameter can not exceed sixteen, and the diameter of its each connected component is at most four and not depending on the dimension $m$ of the system neither on the number $N$ in the control (2).

This conclusion imposes significant limitations on the practical application of linear control. We also note one more drawback of linear control (2): the invariant convex set $A$ of system (1) will not be invariant for system (3).

\subsection{Nonlinear control}

Another type of feedback -- nonlinear -- has the form
\begin{equation}\label{(4)}
{u_n} = - \sum_{j = 1}^{N - 1} {{\varepsilon_j}\left({f({x_{n - jT + T}}) - f({x_{n - jT}})}\right)},
\end{equation}
and the corresponding closed system
\begin{equation}\label{(5)}
{x_{n + 1}} = \sum_{j = 1}^N {{a_j} f\left({{x_{n - jT + T}}}\right)},
\end{equation}
where ${a_1} = 1 - {\varepsilon_1}$, ${a_j} = {\varepsilon_{j - 1}} - {\varepsilon_j}$, $j = 2$, $\ldots$, $N - 1$, ${a_N} = {\varepsilon_{N - 1}}$. It is clear that $\sum_{j = 1}^N {{a_j} = 1}$. Only those controls of the form (4) for which $0 \le {a_j} \le 1$, $ j = 1, \ldots, N$ are considered admissible.

When the state ${x_{k + T}} = {x_k}$, $k = 1, 2, \ldots$ is synchronized, the control (4) vanishes, and the closed system (5) takes the form as in the absence of control. Therefore, the $T$-cycles of the system (1) are $T$-cycles of the system (5). In addition, the invariant convex set $A$ of system (1) remains invariant for system (5).

As shown in [15], for any set $M$ of localization of the multipliers of $T$-cycles of system (1) that does not contain real numbers greater than one, there exists a control of the form (4) for which in system (5) these $T$-cycles will be locally asymptotically stable. Thus, the specified control will have the property of robustness.

We give a solution of the problem of choosing the coefficients ${a_j}$, $j = 1, \ldots, N$, for special cases of the localization sets of multipliers,

case $\mathbf{A}$: $M = \left\{ {\mu \in \mathbb{R}: \mu \in \left({ - \hat \mu, 1}\right)} \right\}$, $\hat \mu > 1$,

case $\mathbf{B}$: $M = \left\{ {\mu \in \mathbb{C}: \left| {\mu + R} \right| < R} \right\}$, $R > {1 \mathord{\left/{\vphantom {1 2}} \right.\kern-\nulldelimiterspace} 2}$.

The algorithm for finding the minimal $N$ and the coefficients $\left\{ {{a_1}, \ldots, {a_N}} \right\}$ consists of the following steps [12]:

a) nodes are calculated:
\[
{\psi_j} = \pi \frac{{\sigma + T(2j - 1)}}{{\sigma + T(N - 1)}},
\]
$j = 1, 2, \ldots, \frac{{N - 2}}{2}$, if $N$ is even; $j = 1, 2, \ldots, \frac{{N - 1}}{2}$, if $N$ is odd; in case $\mathbf{A}$, we should assume $\sigma = 2$, and in case $\mathbf{B}$, assume $\sigma = 1$;

b) the following polynomials are constructed
\[
{\eta_N}\left(z\right) = z\left({z + 1}\right)\prod_{j = 1}^{\frac{{N - 2}}{2}} {\left({z - {e^{i{\psi_j}}}}\right)\left({z - {e^{ - i{\psi_j}}}}\right)},
\]
if $N$ is even;
\[
{\eta_N}\left(z\right) = z\prod_{j = 1}^{\frac{{N - 1}}{2}} {\left({z - {e^{i{\psi_j}}}}\right) \left({z - {e^{ - i{\psi_j}}}}\right)},
\]
if $N$ is odd;

c) the coefficients of the polynomial ${\eta_N}\left(z\right) = \sum_{j = 1}^N {{c_j}{z^j}}$ are calculated (for example, by Vieta's formulas);

d) the coefficients $a_j$ are computed:
\[
{a_j} = \frac{{\left({1 - \frac{{1 + (j - 1)T}}{{2 + (N - 1)T}}}\right){c_j}}}{{\sum_{k = 1}^N {\left({1 - \frac{{1 + (j - 1)T}}{{2 + (N - 1)T}}}\right){c_k}} }}, j = 1, \ldots, N;
\]

e) in case $\mathbf{A}$, we introduce the quantities
\[J_N^{(T)} = - {\left[ {\frac{T}{{2 + (N - 1)T}} \prod_{k = 1}^{\frac{{N - 2}}{2}} {ct{g^2}\frac{{\pi (2 + T(2k - 1))}}{{2(2 + (N - 1)T)}}} } \right]^T}\]
if $N$ is even;
\[J_N^{(T)} = - {\left[ {\prod_{k = 1}^{\frac{{N - 1}}{2}} {ct{g^2}\frac{{\pi (2 + T(2k - 1))}}{{2(2 + (N - 1)T)}}} } \right]^T}\]
if $N$ is odd; the optimal value of $N$ is computed as a minimal natural number satisfying the inequality
$${\mu^*} \le \frac{1}{{\left| {J_{_N}^{(T)}} \right|}};$$

f) in case $\mathbf{B}$ we introduce the quantities
\[\hat J_N^{(T)} = - {\left[ {\frac{T}{{1 + (N - 1)T}} \prod_{k = 1}^{\frac{{N - 2}}{2}} {ct{g^2}\frac{{\pi (1 + T(2k - 1))}}{{2(1 + (N - 1)T)}}} } \right]^T}\]
if $N$ is even;
\[\hat J_N^{(T)} = - {\left[ {\prod_{k = 1}^{\frac{{N - 1}}{2}} {ct{g^2}\frac{{\pi (1 + T(2k - 1))}}{{2(1 + (N - 1)T)}}} } \right]^T}\]
if $N$ is odd; the optimal value of $N$ is computed as a minimal natural number satisfying the inequality 
$$R \le \frac{1}{{2 \left|{\hat J_{_N}^{(T)}} \right|}}.$$

We note that for $\sigma \in \left\{ {1, 2} \right\}$ and $T = 1, 2$, the polynomials ${F_T}\left(z\right) = z {\left({{a_1} + \ldots + {a_N} {z^{N - 1}}}\right)^T}$ are univalent in the central unit disk $D = \left\{ {z \in \mathbb{C}: \left| z \right| < 1} \right\}$. Apparently, the univalence property of polynomials is true for $\sigma \in \left[ {0, 2} \right]$ and for all $T.$ For different $\sigma$ the set $M$ of localization of the multipliers of $T$-cycles of the system (1) must lie in the half-plane $\left\{ {z \in \mathbb{C}: {\mathop{\rm Re}\nolimits} z < 1} \right\}$ (Figure 1-a, 1-b).
\begin{figure}[h]
\begin{minipage}{0.4\textwidth}
\includegraphics[scale=.5]{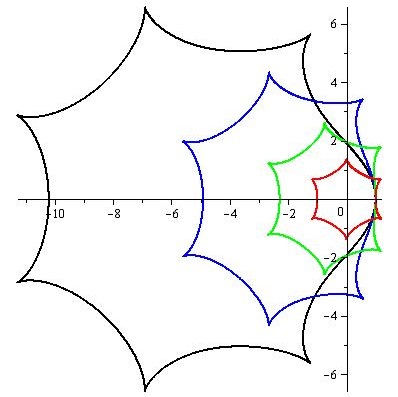}\\ a)
\end{minipage}
\hfil
\begin{minipage}{0.4\textwidth}
\includegraphics[scale=.5]{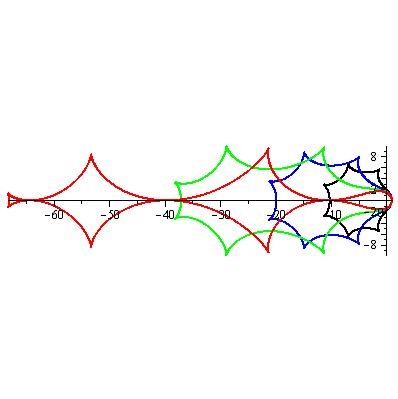}\\ b)
\end{minipage}
\caption{Coverings of the set $M$ of localization of multipliers under
a) $T = 3$, $N = 7$, $\sigma \in \left[ {0, 1} \right]$: $\sigma = 1$ -- black, $\sigma = 0.66$ -- blue, $\sigma = 0.33$ -- green, $\sigma = 0$ -- red;
b) $T = 3$, $N = 7$, $\sigma \in \left[ {1, 2} \right]$: $\sigma = 1$ -- black, $\sigma = 1.33$ -- blue, $\sigma = 1.66$ -- green, $\sigma = 2$ -- red}
\end{figure}

\subsection{Semilinear control}

 To stabilize the cycle of the length $T = 3$ O. Morgul [9, 10]
proposed a feedback  control that includes linear and nonlinear elements, i.e., a semilinear feedback control of the form
\begin{equation}\label{(6)}
{u_n} = - \varepsilon \left({f({x_n}) - {x_{n - T + 1}}}\right),
\end{equation}
for which a corresponding closed-loop system is
\begin{equation}\label{(7)}
x_{n+1}=(1-\varepsilon) f(x) + \varepsilon x_{n-T+1}, 
\end{equation}
where $\varepsilon \in \left[ {0, 1}\right)$. On the cycle, the conditions $f({x_n}) = {x_{n + 1}} = {x_{n - T + 1}}$ are fulfilled, therefore, on the cycle ${u_n} \equiv 0$.

We note that in [9] only the scalar case $f: \mathbb{R} \to \mathbb{R}$ was considered. However, Morgul's scheme can be generalized to the vector case, as will be shown below. In setting the invariant convex set $A$ of system (1) remains invariant for the system (7) too. If we assume that $\varepsilon \in \left[ {0, \infty }\right)$ [9], then the convex invariant set can not be preserved, although in this case it is possible to stabilize the equilibrium positions with multipliers from the half-plane $\left\{z \in \mathbb{C}: {\mathop{\rm Re}\nolimits} z > 1\right\}$.

The characteristic equation for the $T$-cycle, in the scalar case, has the form [10]
\begin{equation}\label{(8)}
\left(\lambda-\varepsilon\right)^T-\mu{\left(1-\varepsilon\right)^T}{\lambda^{T-1}} = 0, 
\end{equation}
where $\mu$ is the multiplier of the cycle. Accordingly, in the vector case the characteristic equation takes the form
\begin{equation}\label{(9)}
\prod_{j = 1}^m {\left[ {\left({\lambda - \varepsilon }\right){ ^T} - {\mu_j} {{\left({1 - \varepsilon }\right)}^T} {\lambda ^{T - 1}}} \right]} = 0, 
\end{equation}
where $\mu_j$ are multipliers of the cycle ($j = 1, \ldots, m$), in general, complex. Equation (9) is obtained as a special case of a more general characteristic equation, which we derive in Section 3.

If all the roots of equation (9) lie in an open central unit disc $D$, then the $T$-cycle is locally asymptotically stable [10, 13]. If the multipliers $\mu { _j}$, $j = 1, \ldots, m$ are known exactly, then one can check whether the roots belong to the central unit circle by known criteria such as Schur-Cohn, Clark, Jury [14]. However, cycles are not known, hence, multipliers are not known. In this case, the geometric criterion of A. Solyanik proved to be effective for the stability of cycles of discrete systems [15]. Let us apply this criterion.

Making the change $\lambda = \frac{1}{z}$, we write equation (9) as a set of equations
\[
\left[\begin{array}{*{20}{c}}
\frac{1}{\mu_j} = \Phi \left(z\right),\\
j = 1, \ldots, m,
\end{array}\right.
\]
where $\Phi \left(z\right) = {\left({1 - \varepsilon }\right)^T} \frac{z}{{{{\left({1 - \varepsilon z}\right)}^T}}}$. The following observation is extremely useful in our settings.

\begin{lemma}\label{lemma1}
All the roots of equation (9) lie in the central unit circle if and only if 
\begin{equation}\label{(10)}
\mu_j\in \left(\overline{\mathbb{C}}\backslash\Phi(\overline{D})\right)^*, \quad j = 1, \ldots, m,
\end{equation}
where $\overline{D} = \left\{z \in \mathbb{C}: \left|z\right|\le 1\right\}$ is a closed central unit disk, $\overline{\mathbb{C}}$ is an extended complex plane, the asterisk denotes the inversion  ${\left(z\right)^*} = \frac{1}{{\bar z}}$. Here $\bar z$ denotes the complex conjugated of $z$.
\end{lemma}

Note that the set ${\left({\bar C\backslash \Phi (\bar D)}\right)^ * }$ is the inverse of the set of exceptional values of the image of the disk under the mapping $\Phi \left(z\right)$. According to Lemma 1, the $T$-cycle will be locally asymptotically stable if the set ${\left({\bar C\backslash \Phi (\bar D)}\right)^ * }$ covers the set $M$ of localization of multipliers. The condition (10) can be rewritten as $M \subseteq {\left({\bar C\backslash \Phi (\bar D)}\right)^ * }$, or in the equivalent form $\Phi (D) \subseteq \bar C\backslash {\left({\bar M}\right)^ * }$. This means that the set ${\left({\bar M}\right)^ * }$ must be exceptional for the image of the disk $D$ under the mapping $\Phi \left(z\right)$.

Let us consider some examples.

\begin{figure}[h]
\includegraphics[scale=.5]{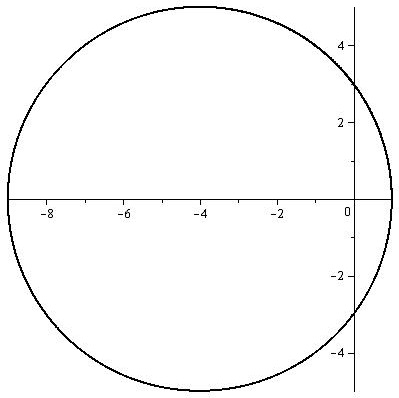}
\caption{Множество ${\left({\bar C\backslash \Phi (\bar D)}\right)^ * }$ for $T = 1$, $\gamma = 0.8$.}
\end{figure}

Example 1. Let $T = 1$. In this case the set
$$
{\left({\bar C\backslash \Phi (\bar D)}\right)^ * } = \left\{ {w \in \mathbb{C}: w = - \frac{\varepsilon }{{1 - \varepsilon }} + \frac{1}{{1 - \varepsilon }} z, z \in D} \right\},
$$
that is, this set is an open circle with center at the point $\left(-\frac{\varepsilon}{1-\varepsilon}, 0\right)$ and of radius $\frac{1}{\left|1 - \varepsilon\right|}$ 
(Fig 2). If $\varepsilon \to 1^-$, the disk converges to the half-plane $\left\{ {w \in \mathbb{C}: w < 1} \right\}$. If $\varepsilon \to 1^+$ the disc converges to the half-plane $\left\{ {w \in \mathbb{C}: w > 1} \right\}$. Therefore, if the set $M$ lies in the half-plane $\left\{ {w \in \mathbb{C}: w < 1} \right\}$ or $\left\{ {w \in \mathbb{C}: w > 1} \right\}$, then the equilibrium position of the system (1) can be stabilized by the control of the form (6).\\

Example 2. Let $T = 2$. Then  
\[
{\left({\bar C\backslash \Phi (\bar D)}\right)^ * } = \left\{ {w \in \mathbb{C}: w = \frac{{2 \varepsilon }}{{{{\left({1 - \varepsilon }\right)}^2}}}\left({\frac{1}{2}(\frac{1}{{\varepsilon z}} + \varepsilon z) - 1}\right), z \in D} \right\},
\]
that is, it is the interior of an ellipse with semiaxes $\left\{ {\frac{{1 + {\varepsilon ^2}}}{{{{(1 - \varepsilon )}^2}}}, \frac{{1 + \varepsilon }}{{1 - \varepsilon }}} \right\}$ and centered at the point $\left({ - \frac{{2 \varepsilon }}{{{{(1 - \varepsilon )}^2}}}, 0}\right)$ (Fig 3). The ellipse foci are at the points $\left({ - \frac{{{{(1 + \varepsilon )}^2}}}{{{{(1 - \varepsilon )}^2}}}, 0}\right)$ and $\left({1, 0}\right)$. Therefore, if the set $M$ lies in the half-plane $\left\{ {w \in \mathbb{C}: w < 1} \right\}$, then the 2-cycle of the system (1) can be stabilized by the control of the form (6 ).

\begin{figure}[h]
\includegraphics[scale=0.6]{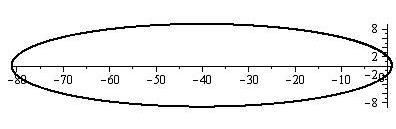}
\caption{The set ${\left({\bar C\backslash \Phi (\bar D)}\right)^ * }$ for $T = 2$, $\gamma = 0.8$}
\end{figure}

O. Morgul considered only the scalar case, and in the scalar case the set $M$ can consist only of real numbers. Then the condition of stabilization of the equilibrium is the following: $\mu \in \left({ - \frac{{1 + \varepsilon }}{{1 - \varepsilon }}, 1}\right)$ or $\mu \in \left({1, \frac{{1 + \varepsilon }}{{1 - \varepsilon }}}\right)$. Accordingly, for the 2-cycle, the stabilizability condition has the form: $\mu \in \left({ - {{\left({\frac{{1 + \varepsilon }}{{1 - \varepsilon }}}\right)}^2}, 1}\right)$.

In the case $T = 1$, the function $\Phi \left(z\right) = (1 - \varepsilon ) \frac{z}{{1 - \varepsilon z}}$ is univalent for all $\varepsilon \in \left({ - \infty, \infty }\right)$ in the entire complex plane, with the exception of the point ${z_0} = \frac{1}{\varepsilon }$. For $T = 2$ and $\varepsilon \ne 0$, the function $\Phi \left(z\right) = {(1 - \varepsilon )^2} \frac{z}{{{{(1 - \varepsilon z)}^2}}}$ is univalent in the open central disk $\left\{ {z \in \mathbb{C}: \left| z \right| < \frac{1}{{\left| \varepsilon \right|}}} \right\}$. In these cases, the functions $\Phi \left(z\right)$ for $\varepsilon \in \left[ {0, 1}\right)$ are univalent in the open central unit disc $D$.

For $T \ge 3$ the situation becomes different. The function $\Phi \left(z\right) = {(1 - \varepsilon )^T} \frac{z}{{{{(1 - \varepsilon z)}^T}}}$ will not be univalent in the disk $D$ for all $\varepsilon \in \left[ {0, 1}\right)$, but only for $\varepsilon \in \left[ {0, \frac{1}{{T - 1}}}\right)$ [16]. Since $\Phi \left({ - 1}\right) = - {\left({\frac{{1 - \varepsilon }}{{1 + \varepsilon }}}\right)^T}$, then for $\varepsilon \in \left[ {0, \frac{1}{{T - 1}}}\right)$, the condition for the stabilizability of the $T$ cycle in the scalar case takes the form $\mu \in \left({ - {{\left({\frac{{1 - \varepsilon }}{{1 + \varepsilon }}}\right)}^T}, 1}\right)$. The function $ {\left({\frac{{1 - \varepsilon }}{{1 + \varepsilon }}}\right)^T}$ increases by $\varepsilon$, hence the maximum size for the multiplier localization set will be $\varepsilon = \frac{1}{{T - 1}}$, i.e. $\mu \in \left({ - {{\left({\frac{T}{{T - 2}}}\right)}^T}, 1}\right)$. For $\varepsilon > \frac{1}{{T - 1}}$, the function $\Phi \left(z\right)$ fails to be univalent, and the interval for the multiplier will decrease (Figure 4).

The function $ {\left({\frac{T}{{T - 2}}}\right)^T}$ decreases as $T \ge 3$ asymptotically tending to ${e^2} \approx 7.389$.

\begin{figure}[h]
\begin{minipage}{0.45\textwidth}
\includegraphics[scale=0.45]{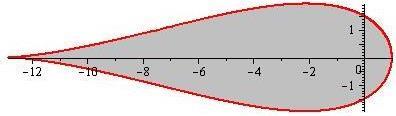}\\ 
a)
\end{minipage}
\hfil
\begin{minipage}{0.45\textwidth}
\includegraphics[scale=0.45]{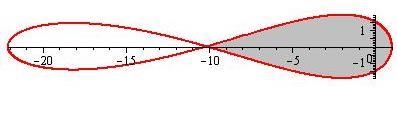}\\b)
\end{minipage}
\caption{The inverse image of the boundary of the disk $D$ (red) and the set ${\left({\bar C\backslash \Phi (\bar D)}\right)^*}$ (gray) for $T = 5$, a) $\varepsilon = 0.25$, b) $\varepsilon = 0.3$.}
\end{figure}

\section{Generalized semilinear control}

\subsection{Formulation of the problem}

A natural generalization of delayed feedback control is the joint use of the linear, nonlinear, and semilinear feedback considered in Section 2. Let us take into account that the linear feedback is ineffective, and we choose the control in the form of a convex combination of nonlinear and generalized semilinear feedbacks, specifically,
\begin{equation}\label{(11)}
{u_n} = - \left({1 - \gamma }\right) \sum_{j = 1}^{N - 1} {\varepsilon_j^{(1)}\left({f({x_{n - jT + T}}) - f({x_{n - jT}})}\right) } - \gamma \sum_{j = 1}^N {\varepsilon_j^{(2)} \left({f({x_{n - jT + T}}) - {x_{n - jT + 1}})}\right)},
\end{equation}
where $\gamma \in \left[ {0, 1}\right)$. Note that the control (11) disappears on a cycle of length $T$.

The motivation for using a control of the form (11) for $T \ge 3$ is completely obvious. In the general case, the semilinear control does not allow stabilizing cycles of system (1) of length three or more. However, the combined use of semilinear and non-linear controls can allow reducing the necessary length of history used in the feedback.

For $T = 1, 2$, we can also expect qualitatively new effects when the equilibrium position is stabilized due to a larger number of control parameters: an increase in the rate of convergence of the perturbed solutions to periodic ones, expansion of the basin of attraction of a locally stable periodic solution, and so on. In other words, combined control should improve the properties of both nonlinear and semilinear control.

Let us close the system (1) by the control (11), then we get
\begin{equation}\label{(12)}
{x_{n + 1}} = \left({1 - \gamma }\right)\sum_{j = 1}^N {{a_j}f({x_{n - j T + T}})} + \gamma \sum_{j = 1}^N {{b_j}{x_{n - j T + 1}}},
\end{equation}
where the coefficients ${a_1}, \ldots, {a_N}$, ${b_1}, \ldots, {b_N}$ are associated with the parameters $\varepsilon_1^{(1)}, \ldots, \varepsilon_{N - 1}^{(1)}$, $\varepsilon_1^{(2)}, \ldots, \varepsilon_N^{(2)}$ by a linear bijection
\[
\left\{ {\begin{array}{*{20}{c}}
{{a_1} = \frac{1}{{1 - \gamma }} - \varepsilon_1^{(1)} - \frac{\gamma }{{1 - \gamma }}\varepsilon_1^{(2)}, }\\
{{a_j} = - (\varepsilon_j^{(1)} - \varepsilon_{j - 1}^{(1)}) - \frac{\gamma }{{1 - \gamma }} \varepsilon_j^{(2)}, j = 2, \ldots, N - 1,}\\
{{a_N} = \varepsilon_{N - 1}^{(1)} - \frac{\gamma }{{1 - \gamma }}\varepsilon_N^{(2)}, }\\
{{b_j} = \varepsilon_j^{(2)}, j = 1, \ldots, N. }
\end{array}} \right.
\]
We request the invariant convex set $A$ of system (1) to be invariant for system (12). Therefore, we must require the following relations: ${a_j} \in \left[ {0, 1} \right]$, ${b_j} \in \left[ {0, 1} \right]$, $j = 1, \ldots, N$, $\sum_{j = 1}^N {{a_j}} = 1$, $\sum_{j = 1}^N {{b_j}} = 1$. To this end, additional restrictions must be imposed on the control (11). Namely, $\sum_{j = 1}^N {\varepsilon_j^{(2)}} = 1$; $\frac{1}{{1 - \gamma }} - \varepsilon_1^{(1)} \ge \frac{\gamma }{{1 - \gamma }}\varepsilon_1^{(2)} \ge 0$, $\varepsilon_{j - 1}^{(1)} - \varepsilon_j^{(1)} \ge \frac{\gamma }{{1 - \gamma }} \varepsilon_j^{(2)} \ge 0$, $j = 2, \ldots, N - 1$, $\varepsilon_{N - 1}^{(1)} \ge \frac{\gamma }{{1 - \gamma }}\varepsilon_N^{(2)} \ge 0$.

It is required to select the parameters ${a_1}, \ldots, {a_N}$, ${b_1}, \ldots, {b_N}$, satisfying the given constraints, so that the $T$-cycle of the system (12) is locally asymptotically stable, and $N$ would be the smallest.

If we let $\gamma = 0$ in (12), then we obtain the system (5), i.e. the system (1), closed by nonlinear feedback. If we let $N = 1$ in (12), then ${a_1} = {b_1} = 1$, therefore, we get the system (7), that is, as in the case of closure by the semilinear feedback. Thus, the system (12) contains the systems (5) and (7) as particular cases.

\subsection{Construction of the characteristic polynomial}

We begin the investigation of the stability of the $T$ -cycle of the system (12) with the derivation of the characteristic equation for this cycle. The classical way is the construction of the Jacobi matrix of a special mapping in the neighborhood of the cycle [10], and finding the characteristic polynomial of this matrix. As a result, this characteristic polynomial will have a cumbersome form, and the path of its simplification is not at all obvious [10].

The same polynomial can be constructed from a different mapping, and the polynomial is obtained in a very convenient form for further investigations [17, 18]. In [18] the equivalence of the classical O. Morgul method and the alternative method is proved, which will be applied below.

The solution of system (12) can be represented in the form
\begin{equation}\label{(13)}
\left\{ {\begin{array}{*{20}{c}}
{{x_{T s}} = {\eta_1} + u_s^1}\\
{{x_{T s + 1}} = {\eta_2} + u_s^2}\\
\ldots \\
{{x_{T s + T - 1}} = {\eta_T} + u_s^T}
\end{array}} \right.,
\end{equation}
$s = 0, 1, \ldots$. We substitute solution (13) into (12), assuming that in the neighborhood of the cycle the norms of the vectors $u_s^1, \ldots, u_s^T$ are small.

Let $n = T s$. Then
\[
{x_{n + 1}} = {x_{Ts + 1}} = {\eta_2} + u_s^2, {x_{n + 2}} = {x_{Ts + 2}} = {\eta_3} + u_s^3, \ldots, {x_{n + T}} = {x_{T(s + 1)}} = {\eta_1} + u_{s + 1}^1.
\]
Selecting the linear part and taking into account that ${\eta_1} = f\left({{\eta_2}}\right), \ldots, {\eta_T} = f\left({{\eta_1}}\right)$, we get
\begin{equation}\label{(14)}
\left\{ {\begin{array}{*{20}{c}}
{u_s^2 = (1 - \gamma ) f'\left({{\eta_1}}\right)({a_1}u_s^1 + \ldots + {a_N}u_{s - N + 1}^1) + \gamma ({b_1}u_{s - 1}^2 + \ldots + {b_N}u_{s - N}^2)}\\
\ldots \\
{ u_s^T = (1 - \gamma ) f'\left({{\eta_{T - 1}}}\right)({a_1}u_s^{T - 1} + \ldots + {a_N}u_{s - N + 1}^{T - 1}) + \gamma ({b_1}u_{s - 1}^T + \ldots + {b_N}u_{s - N}^T)}\\
{u_{s + 1}^1 = (1 - \gamma ) f'\left({{\eta_T}}\right)({a_1}u_s^T + \ldots + {a_N}u_{s - N + 1}^T) + \gamma ({b_1}u_s^1 + \ldots + {b_N}u_{s - N + 1}^1)}
\end{array}} \right.,
\end{equation}
where $f'\left({{\eta_j}}\right)$, $j = 1, \ldots, T$, are Jacobian matrices of dimension $m \times m$.

The system (14) is linear, so its solutions are represented in the form
\begin{equation}\label{(15)}
\left({\begin{array}{*{20}{c}}
{u_s^1}\\
{\ldots}\\
{u_s^T}
\end{array}}\right) = \left({\begin{array}{*{20}{c}}
{{c_1}}\\
{\ldots}\\
{{c_T}}
\end{array}}\right) {\lambda ^s},
\end{equation}
Where $\lambda$ is a complex number to be determined. Substituting (15) into (14), we obtain
\begin{equation}\label{(16)}
\left\{\begin{aligned}
-(1-\gamma)(a_1\lambda^s + \ldots + a_N\lambda^{s-N+1}) f'\left(\eta_1\right)c_1 + (\lambda^s-\gamma (b_1\lambda^{s - 1} + \ldots + b_N\lambda^{s-N}))c_2 &= 0\\
\ldots \\
- (1 - \gamma ) ({a_1}{\lambda ^s} + \ldots + {a_N}{\lambda ^{s - N + 1}})f'\left({{\eta_{T - 1}}}\right) {c_{T - 1}} + ({\lambda ^s} - \gamma ({b_1}{\lambda ^{s - 1}} + \ldots + {b_N}{\lambda ^{s - N}})) {c_T} &= 0\\
- (1 - \gamma ) ({a_1}{\lambda ^s} + \ldots + {a_N}{\lambda ^{s - N + 1}})f'\left({{\eta_T}}\right) {c_T} + ({\lambda ^{s + 1}} - \gamma ({b_1}{\lambda ^s} + \ldots + {b_N}{\lambda ^{s - N + 1}})) {c_1} &= 0
\end{aligned}\right.
\end{equation}
Denote by $z = \frac{1}{\lambda }$, $q\left(z\right) = {a_1} + {a_2}z + \ldots + {a_N}{z^{N - 1}}$, $p\left(z\right) = {b_1} z + {b_2} {z^2} + \ldots + {b_N}{z^N}$. The system (16) considered with respect to the vectors ${c_1}, \ldots, {c_T}$, will have a nontrivial solution if and only if the determinant of the matrix
{\tiny
\[
\left({\begin{array}{*{20}{c}}
{ - (1 - \gamma ) q(z)f'\left({{\eta_1}}\right)}&{(1 - \gamma p(z))I}&{\rm O}& \ldots &{\rm O}&{\rm O}\\
{\rm O}&{ - (1 - \gamma ) q(z)f'\left({{\eta_2}}\right)}&{(1 - \gamma p(z))I}& \ldots &{\rm O}&{\rm O}\\
\ldots & \ldots & \ldots & \ldots & \ldots & \ldots \\
{\rm O}&{\rm O}&{\rm O}& \ldots &{ - (1 - \gamma ) q(z)f'\left({{\eta_{T - 1}}}\right)}&{(1 - \gamma p(z))I}\\
{{z^{ - 1}} (1 - \gamma p(z))I}&{\rm O}&{\rm O}& \ldots &{\rm O}&{ - (1 - \gamma ) q(z)f'\left({{\eta_T}}\right)}
\end{array}}\right),
\]
}
is non-zero, where ${\rm O}$ is zero matrix of dimension $m \times m$, $I$ is the identity matrix of dimension $m \times m$. That is
\[
\det \left(z^{-1}(1 - \gamma p(z))^T I - \left((1-\gamma)q\left(z\right)\right)^T\prod_{j = 1}^T f'\left(\eta_j\right)\right) = 0.
\]

Let the eigenvalues of the product of the Jacobi matrices $\prod_{j = 1}^T {f'\left({{\eta_j}}\right)}$ be equal to ${\mu_1}, \ldots, {\mu_m}$. Then, replacing this product by Jordan's canonical form, we obtain the final form of the characteristic equation
\begin{equation}\label{(17)}
\prod_{j = 1}^m {\left({\frac{{{{(1 - \gamma p(z))}^T}}}{{z{{\left({(1 - \gamma )q\left(z\right)}\right)}^T}}} - {\mu_j}}\right)} = 0.
\end{equation}
Hence, the desired characteristic polynomial has the form
\begin{equation}\label{(18)}
\tilde f\left(\lambda\right) = \prod_{j = 1}^m {\left({ {{\left[ {{\lambda ^N} - \gamma {\lambda ^N} p({\lambda ^{ - 1}})} \right]}^{ T}} - {{(1 - \gamma )}^T}{\mu_j} {\lambda ^{T - 1}}{{\left[ {{\lambda ^{N - 1}} q({\lambda ^{ - 1}})} \right]}^{ T}}}\right)}.
\end{equation}
The polynomial (18) contains, as a special case for $N = 1$, the polynomial (9).

\subsection{Geometric criterion for local asymptotic stability of a cycle}

The next step in the study of the stability of cycles is to analyze the location of the zeros of the characteristic polynomial (18) on the complex plane. Or, equivalently, the roots of equation (17). The local stability of cycles of difference systems is equivalent to the Schur stability of the characteristic polynomial corresponding to this cycle [cf. Ex. 13]. We present this fact in the following two Lemmas.

\begin{lemma}
The $T$-cycle of system (12) is locally asymptotically stable if and only if all zeros of polynomial (18) lie in the open central unit disk $D$.
\end{lemma}

As noted in Section 2.3, it is not possible to apply the known criteria for testing the Schur stability of the polynomial (18), since the quantities ${\mu_1}, \ldots, {\mu_m}$ are not known. Therefore, to verify the local stability of cycles of system (12), we apply the geometric criterion of stability suggested by A. Solyanik. Denote by $\Phi\left(z\right) = \left(1-\gamma\right)^T\frac{z\left(q(z)\right)^T}{\left(1-\gamma p(z) \right)^T}$.

\begin{lemma}
All roots of the polynomial (18) lie in the open central unit disc $D$ if and only if 
\begin{equation}\label{(19)}
{\mu_j} \in {\left({\bar C\backslash \Phi (\overline D )}\right)^*}, j = 1, \ldots, m,
\end{equation}
where $\overline{D}$ is a closed central unit disk, $\overline{\mathbb{C}}$ is an extended complex plane, the asterisk denotes the inversion operation: $\left(z\right)^* =\frac{1}{\bar z}$.
\end{lemma}
\begin{proof}
The polynomial (18) is Shur-stable if and only if $\tilde f(\lambda ) \ne 0$ for all $\lambda \in \bar C\backslash D$. This is equivalent to $\frac{1}{{{\mu_j}}} \ne \Phi (z)$, $z \in \bar D$, $j = 1, \ldots, m$. Consequently, the necessary and sufficient conditions for the stability of the Schur polynomial (18) are the inclusions: $\frac{1}{{{\mu_j}}} \notin \Phi (\bar D)$, or $\frac{1}{{{\mu_j}}} \in \bar C\backslash \Phi (\bar D)$, or ${\mu_j} \in {\left({\bar C\backslash \Phi (\overline D )}\right)^ * }$, $j = 1, \ldots, m$.
\end{proof}

In the general case, cycle multipliers are not known; hence, the $T$ -cycle is locally asymptotically stable if the set ${\left({\bar C\backslash \Phi (\bar D)}\right)^ * }$ covers the set $M$ of localization of multipliers. This means that the set ${\left({\bar M}\right)^ * }$ must be exceptional for the image of the disk $D$ under the mapping $\Phi \left(z\right)$. This property will be the main one for constructing the control coefficients ${a_1}, \ldots, {a_N}, {b_1}, \ldots, {b_N}$.

\subsection{The design of controls that stabilize cycles}

The next step is to construct a function $\Phi \left(z\right)$ so that the set ${\left({\bar C\backslash \Phi (\bar D)}\right)^ * }$ covers the set $M$ of localization of multipliers. In this case, it is necessary to estimate the size of the set ${\left({\bar C\backslash \Phi (\bar D)}\right)^ * }$ as a function of $N$ and $\gamma$. The function $\Phi \left(z\right) = {\left(1 - \gamma\right)^T} \frac{z\left(q(z)\right)^T}{\left(1 - \gamma p(z)\right)^T}$ depends on the polynomials $q(z)$, $p(z)$, and the parameter $\gamma \in \left[ {0, 1}\right)$, and $q(1) = 1$, $p(0) = 0$, $p(1) = 1$, hence $\Phi (0) = 0$, $\Phi (1) = 1$.

To further advance the formulation of the problem, we impose an essential restriction on the function $\Phi \left(z\right)$, namely, we know that the polynomial $q(z)$ is calculated by the formulas indicated in Section 2.2 (for some $\sigma \in \left[ {1, 2} \right]$). That ensures that $M \subseteq \left\{ {\mu \in \mathbb{R}: \mu \in \left({ - \hat \mu, 1}\right)} \right\}$ ($\hat \mu > 1$) or $M \subseteq \left\{ {\mu \in \mathbb{C}: \left| {\mu + R} \right| < R} \right\}$ ($R > {1 \mathord{\left/{\vphantom {1 2}} \right.\kern-\nulldelimiterspace} 2}$) for any admissible $\hat \mu$ and $R$, at least for a sufficiently large $N$ and $\gamma = 0$. We want to choose the polynomial $p(z)$ and the parameter $\gamma$ so that a certain desired linear dimension of the set ${\left({\bar C\backslash \Phi (\bar D)}\right)^ * }$ is maximal (or the set $\Phi \left(D\right)$ is minimal). The linear dimensions depend on the value $\Phi \left({ - 1}\right) = - {\left({1 - \gamma }\right)^T} \frac{{{{\left({q( - 1)}\right)}^T}}}{{{{\left({1 - \gamma p( - 1)}\right)}^T}}}$. Since all ${b_j}$, $j = 1, \ldots, N$ are not negative, then $\left| {p( - 1)} \right| < 1$. Consequently, it is not possible to make $\Phi \left({ - 1}\right)$ small at the expense of the polynomial $p(z)$. The polynomial $p(z)$ has a very specific role: this polynomial should be chosen so that the parameter $\gamma$ can be varied within the widest limits.

We formulate this requirement. We consider the family of functions
\begin{equation}\label{(20)}
\left\{\Phi \left(z\right) = \left(1-\gamma\right)^T\frac{z\left(q(z)\right)^T}{\left({1-\gamma p(z)}\right)^T}: \gamma \in \left[0, \gamma^*\right]\right\}, 
\end{equation}
where the polynomial $q(z)$ is defined as above. It is required to find the polynomial $p(z)$ (with a given degree and given normalization conditions) so that the family of functions (20) is univalent in the disc $D$, and $\gamma^*$ is maximal.

If it is necessary to maximize the linear dimension of the set ${\left({\bar C\backslash \Phi (\bar D)}\right)^ * }$ in the direction of the negative real axis, then the requirement of univalence of the family (20) can be replaced by a weaker requirement to be typically real. We recall that an analytic function in $D$ is said to be typically real in the sense of Rogozinsky if real preimages correspond to real values of the function [19]. In other words, a function that is typically real in $D$ must map an open upper semicircle to an open upper (or lower) half-plane.

We give a solution of this problem for $T = 1$. The function $\Phi \left(z\right) = \left({1 - \gamma }\right)\frac{{z }}{{1 - \gamma z}}$ is univalent in $D$ for $\gamma \in \left[ {0, 1}\right)$. The polynomial $z q(z)$ is also univalent for $z \in D$. Therefore, the function $\Phi \left(z\right) = \left({1 - \gamma }\right)\frac{{z q(z)}}{{1 - \gamma z q(z)}}$ is univalent for $\gamma \in \left[ {0, 1}\right)$ and $z \in D$ as a superposition of univalent functions.

In this case, the set ${\left({\bar C\backslash \Phi (\bar D)}\right)^ * } = \left\{ {\frac{1}{{1 - \gamma }}\left({\frac{1}{{zq(z)}} - \gamma }\right): z \in D} \right\}$. This set is obtained as a result of shifting the set $\left\{\frac{1}{zq(z)}: z \in D\right\}$ by $\gamma$ and the subsequent extension in $\frac{1}{1-\gamma}$ times (Fig 5).

\begin{figure}[h]
\includegraphics[scale=.7]{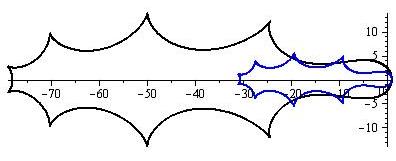}
\caption{The sets ${\left({\bar C\backslash \Phi (\bar D)}\right)^ * }$ for $N = 9$, $\sigma = 1.8$, $\gamma = 0$ -- blue, $\gamma = 0.6$ -- black.}
\end{figure}

The system (12) takes the form
\begin{equation}\label{(21)}
{x_{n + 1}} = \left({1 - \gamma }\right) \sum_{j = 1}^N {{a_j}f({x_{n - j + 1}}) } + \gamma \sum_{j = 1}^N {{a_j}{x_{n - j + 1}}}.
\end{equation}
The role of the parameter $\gamma$ is clearly visible in Fig 5. The value of $N$ plays a dual role with respect to the parameter $\gamma$. Increasing $N$ can reduce $\gamma$, leaving the linear dimensions of the set ${\left({\bar C\backslash \Phi (\bar D)}\right)^ * }$ almost unchanged. We consider the system (12) for $N = 1$:
\begin{equation}\label{(22)}
{x_{n + 1}} = \left({1 - \gamma }\right) f({x_n}) + \gamma {x_n}.
\end{equation}
The boundary of the set ${\left({\bar C\backslash \Phi (\bar D)}\right)^ * }$ is a circle passing through the point $\left({ - \frac{{1 + \gamma }}{{1 - \gamma }}, 0}\right)$. Consider system (12) with other averaging parameters  ${\gamma_1}$, ${a_1}, \ldots, {a_N}$
\begin{equation}\label{(23)}
{x_{n + 1}} = \left({1 - {\gamma_1}}\right) \sum_{j = 1}^N {{a_j}f({x_{n - j + 1}}) } + {\gamma_1} \sum_{j = 1}^N {{a_j}{x_{n - j + 1}}}.
\end{equation}
For this system, the boundary of the set ${\left({\bar C\backslash \Phi (\bar D)}\right)^ * }$ passes through the point $\left({ - \frac{{q_N^{ - 1} + {\gamma_1} }}{{1 - {\gamma_1}}}, 0}\right)$, where ${q_N} = - \sum_{j = 1}^N {{{( - 1)}^j}{a_j}}$, and the coefficients ${a_1}, \ldots, {a_N}$ are calculated using the formulas of Section 2.2 for some $\sigma \in \left[ {1, 2} \right]$. Let $q_N^{ - 1} < \frac{{1 + {\gamma_1} }}{{1 - {\gamma_1}}}$. This means that the linear dimension of the set ${\left({\bar C\backslash \Phi (\bar D)}\right)^ * }$ for the system (22) is greater than for the system (23). In order for the linear dimensions of these sets to be almost equal, the second set must be stretched. The stretching coefficient is determined by the parameter ${\gamma_1}$. It is not difficult to establish a connection between the parameters $\gamma$ and ${\gamma_1}$:
\[
{\gamma_1} = \frac{{1 - q_N^{ - 1}}}{2} + \frac{{1 + q_N^{ - 1}}}{2}\gamma.
\]
Let, for example, $\gamma = 0.9$, $N = 5$, $\sigma = 1.0$. Calculate $q_N^{ - 1} \approx 5.0$. Then ${\gamma_1} \approx 0.7$. This example shows how much the ${\gamma_1}$ parameter can be made smaller compared to $\gamma$. The sets ${\left({\bar C\backslash \Phi (\bar D)}\right)^ * }$ for systems (22), (23) are shown in Fig. 6.

\begin{figure}[h]
\includegraphics[scale=.7]{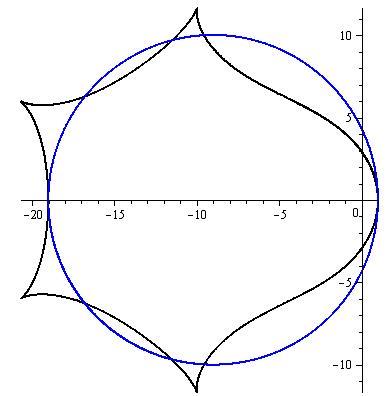}
\caption{The sets ${\left({\bar C\backslash \Phi (\bar D)}\right)^ * }$ for the systems (22) (blue) and (23) (black) for $N = 5$, $\sigma = 1.0$, $\gamma = 0.9$, ${\gamma_1} \approx 0.7$.}
\end{figure}

Let us give some more examples for $\sigma = 1.4$, $\sigma = 1.8$, $\sigma = 2.0$ and $N = 5$.

If $\sigma = 1.4$, then $q_N^{ - 1} \approx 7.856$ and ${\gamma_1} \approx 0.557$. If $\sigma = 1.8$, then $q_N^{ - 1} \approx 11.640$ and ${\gamma_1} \approx 0.368$. If $\sigma = 2.0$, then $q_N^{ - 1} \approx 13.928$ and ${\gamma_1} \approx 0.254$.

\begin{figure}[h]
\begin{minipage}{.3\textwidth}
\includegraphics[scale=.35]{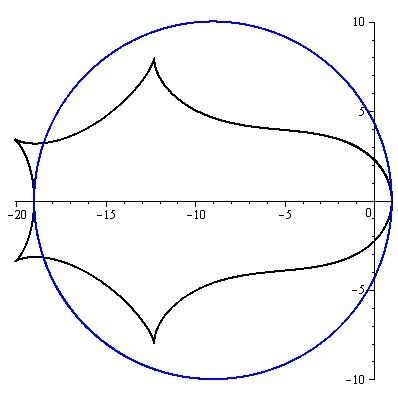}\\
\centering a)
\end{minipage}
\hfil
\begin{minipage}{.3\textwidth}
\includegraphics[scale=.35]{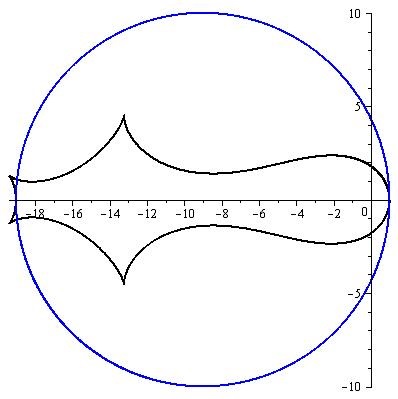}\\
\centering b)
\end{minipage}
\hfil
\begin{minipage}{.3\textwidth}
\includegraphics[scale=.35]{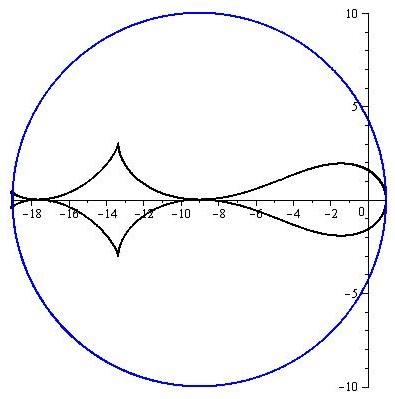}\\
\centering c)
\end{minipage}
\caption{The sets ${\left({\bar C\backslash \Phi (\bar D)}\right)^ * }$ for the systems (22) (blue) and (23) (black) for a) $N = 5$, $\sigma = 1.4$, $\gamma = 0.9$, ${\gamma_1} \approx 0.557$; b) $N = 5$, $\sigma = 1.8$, $\gamma = 0.9$, ${\gamma_1} \approx 0.368$ c) $N = 5$, $\sigma = 2.0$, $\gamma = 0.9$, ${\gamma_1} \approx 0.254$.}
\end{figure}

Thus, the introduction of retardation in the feedback allows us to reduce the value of the averaging parameter~$\gamma$. The advantages of this approach are discussed below.

\subsection{On the rate of convergence of perturbed solutions to the equilibrium position}

Consider again the system (22), and let $M$ be the localization set of the multipliers of the equilibrium position of this system. Let the control parameters $\gamma, {a_1}, \ldots, {a_N}$ be chosen so that the domain ${\left({\bar C\backslash \Phi (\bar D)}\right)^ * }$ covers the set $M$, where $\Phi \left(z\right) = \left({1 - \gamma }\right)\frac{{z q(z)}}{{1 - \gamma z q(z)}}$, and $q(z) = {a_1} + {a_2}z + \ldots + {a_N}{z^{N - 1}}$. In this case, the equilibrium position will be locally asymptotically stable. This means that for the initial vectors $ {x_1}, \ldots, {x_N}$ lying in a sufficiently small neighborhood of the equilibrium position, the solution of the system (22) determined by these initial vectors tends to the equilibrium position. Such a neighborhood is called the basin of attraction of the equilibrium position of the system (22) in the space of the initial vectors. Evaluation of the basin of attraction is, in general, a very complicated task, and it is not the subject of this article.

We note, however, that even when all the conditions for attraction of the perturbed solution to the equilibrium position are satisfied, the behavior of the perturbed solution may turn out to be complicated, and approach the equilibrium very slowly. This is the case when the multiplier of the system is near the boundary of the set ${\left({\bar C\backslash \Phi (\bar D)}\right)^ * }$. The rate at which the perturbed solution approaches the equilibrium is determined by the maximum ${\lambda ^ * }$ among the moduli of zeros of the characteristic polynomial (18) (with $q(z) = {a_1} + {a_2}z + \ldots + {a_N}{z^{N - 1}}$, $p(z) = z q(z)$, $z = \lambda^{-1}$).

Constructing the $\Phi \left({\frac{1}{\rho }{e^{i t}}}\right)$ maps as $\rho \le 1$, one can obtain the level lines ${\lambda ^ * } = \rho$. Fig 8 and 9 show these level lines for the polynomial (9) for $\gamma = 0.9$ and the polynomial (18) for $\gamma = 0.9$, $N = 5$, $\sigma \in \left\{ {1, 1.4, 1.8, 2} \right\}$. A darker color shows the level lines corresponding to a larger value of $\rho$.

\begin{figure}[h]
\includegraphics[scale=.7]{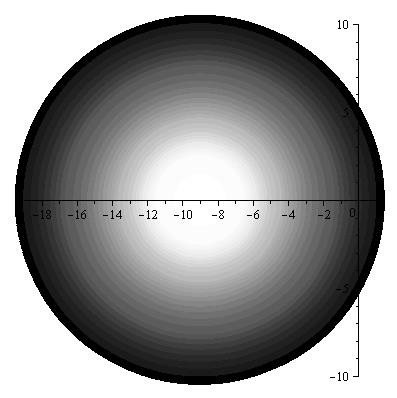}
\caption{The set ${\left({\bar C\backslash \Phi (\bar D)}\right)^ * }$, where $\Phi \left(z\right) = \left({1 - \gamma }\right)\frac{{z }}{{1 - \gamma z }}$, $\gamma = 0.9$.}
\end{figure}

\begin{figure}[h]
\begin{minipage}{.45\textwidth}
\includegraphics[scale=.6]{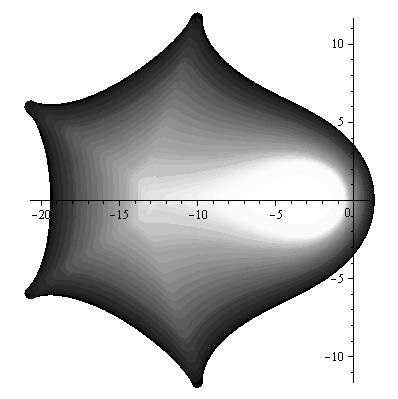}\\
\centering a)
\end{minipage}
\hfil
\begin{minipage}{.45\textwidth}
\includegraphics[scale=.6]{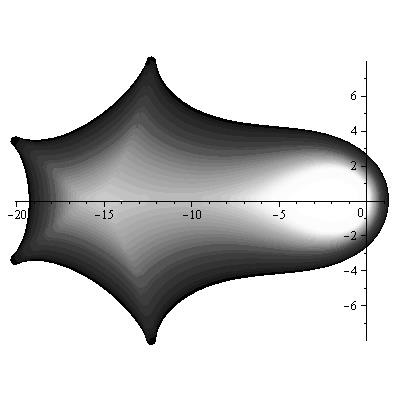}\\
\centering b)
\end{minipage}

\begin{minipage}{.45\textwidth}
\includegraphics[scale=.6]{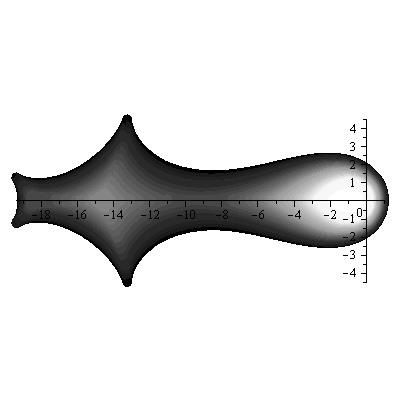}\\
\centering c)
\end{minipage}
\hfil
\begin{minipage}{.45\textwidth}
\includegraphics[scale=.6]{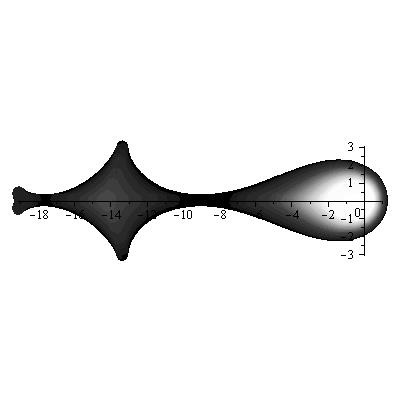}\\
\centering d)
\end{minipage}

\caption{The set ${\left({\bar C\backslash \Phi (\bar D)}\right)^ * }$, where $\Phi \left(z\right) = \left({1 - \gamma }\right)\frac{{z q(z) }}{{1 - \gamma z q(z) }}$, $N = 5$, for a) $\gamma = 0.7$, $\sigma = 1$, b) $\gamma \approx 0.557$, $\sigma = 1.4$, c) $\gamma \approx 0.368$, $\sigma = 1.8$, d) $\gamma \approx 0.254$, $\sigma = 2$.}
\end{figure}

In Fig. 8, 9, the darker regions correspond to those values of the multipliers $\mu$ for which the maximum ${\lambda^*}$ among the moduli of zeros of the characteristic polynomial is closer to unity.

Let us consider in more detail the diagrams shown in Fig. 8 and 9. The light regions determine the effective coverage of the localization set of the multipliers of the equilibrium position of the system (1). If we use the semi-linear control by O. Morgul, then the set ${\left({\bar C\backslash \Phi (\bar D)}\right)^ * }$ can theoretically be made arbitrarily large, letting $\gamma$ to unity , covering arbitrarily large regions of localization of multipliers. However, in this case, the effective coverage region shifts to a neighborhood of the point $\left({ - \frac{\gamma }{{1 - \gamma }}, 0}\right)$ i.e. it is significantly shifted from zero. Thus, if it is necessary to cover multipliers, one of which in the unit circle and the other on the negative real axis and at a considerable distance from zero, the first of them will appear in the ``dark" region, and therefore the eigenvalue that corresponds to it, lies close to the boundary of the unit circle. This, in turn, means a very slow aspiration of the perturbed solution to the equilibrium position.

When generalized semilinear control is used, the effective coverage area of the localization set of the multipliers is sufficiently close to zero and extends to the negative real axis at $\gamma \to 1$ or $N \to \infty$. Also from Fig. 9, the role of the parameter $\sigma$ is seen. Thus, the use of generalized semilinear control makes it possible to accelerate the convergence of perturbed solutions to the equilibrium position in comparison with the control by O.~Morgul. Especially in the case of a large spread of multipliers, the equilibrium positions of the system (1).

\section{Applications to computational methods for solving systems of equations}

In this section, we consider several examples of the application of the method of stabilizing the equilibrium position of the system (1) by control (11) to the possibility of generalizing the known iterative processes for solving systems of linear and nonlinear equations [20].

\subsection{Nonlinear equations}

Consider the computational scheme of the method of simple iterations (or the Richardson method) of solving a system of algebraic equations, generally speaking with complex coefficients
\begin{equation}\label{(24)}
F\left(x\right) = 0,
\end{equation}
where the differentiable function $F: {C^m} \to {C^m}$. To solve the system (24) an auxiliary difference system is constructed
\begin{equation}\label{(25)}
{x_{n + 1}} = {x_n} + G\left({x_n}\right)F\left({x_n}\right),
\end{equation}
where $G\left({x_n}\right)$ is a matrix to be chosen. The equilibrium positions of system (25) coincide with the solutions of system (24). In the classical scheme of simple iterations, the matrix $G\left({x_n}\right)$ is chosen from the condition that the multipliers of the equilibrium position of the system (25) belong to the interval $\left({ - 1, 1}\right)$. This condition can be weakened: the matrix $G\left({x_n}\right)$ should be chosen so that the multipliers of the equilibrium position of the system (25) are real and less than unity. For example, we can take $G\left({x_n}\right) = - {\left[ {F'({x_n})} \right]^ * }$, where $F\left(x\right)$ is the Jacobian matrix, the sign $*$ means Hermitian transposition. Then the system (25) takes the form
\begin{equation}\label{(26)}
{x_{n + 1}} = {x_n} - {\left[ {F'\left({x_n}\right)} \right]^ * }F\left({x_n}\right).
\end{equation}
Let $F\left(\xi\right) = 0$. If the matrix $F'(\xi )$ is not degenerated, then the matrix ${\left[ {F'(\xi )} \right]^ * }F'(\xi )$ is positive definite, i.e all its eigenvalues are greater than zero. Consequently, all the eigenvalues of the matrix $\left({I - {{\left[ {F'(\xi )} \right]}^ * }F'(\xi )}\right)$, where $I$ is unit matrix, are real and less than one. Let these eigenvalues lie in the interval $\left({ - \hat \mu, 1}\right)$.

We organize the iteration process for system (26) according to the scheme (12):
\begin{equation}\label{(27)}
{x_{n + 1}} = \sum_{j = 1}^N {{a_j}{x_{n - j + 1}}} - (1 - \gamma )\sum_{j = 1}^N {{a_j}{{\left[ {F'\left({{x_{n - j + 1}}}\right)} \right]}^ * }F\left({{x_{n - j + 1}}}\right)},
\end{equation}
where $0 < \gamma < 1$, the coefficients ${a_1}, \ldots, {a_N}$ are calculated using the formulas of Section 2.2 (for some $\sigma \in \left[ {1, 2} \right]$). Denote by ${q_N} = - \sum_{j = 1}^N {{{( - 1)}^j}{a_j}}$. Then $\gamma$ and $N$ should be chosen from the conditions: $\frac{{q_N^{ - 1} + \gamma }}{{1 - \gamma }} > \hat \mu$, $0 < \gamma < 1$.

For example, if $\sigma = 2$, then ${a_j} = 2\tan \frac{\pi }{{2(N + 1)}}\left({1 - \frac{j}{{N + 1}}}\right)\sin \frac{{\pi j}}{{N + 1}}, j = 1, \ldots, N$, ${q_N} = {\tan ^2}\frac{\pi }{{2(N + 1)}}$, and the inequality $\frac{{{{\cot }^2}\frac{\pi }{{2(N + 1)}} + \gamma }}{{1 - \gamma }} > \hat \mu$ must hold.

If $\sigma = 1$, then ${a_j} = \frac{2}{N}\left({1 - \frac{j}{{N + 1}}}\right), j = 1, \ldots, N$, ${q_N} = \frac{1}{N}$, and the inequality $\frac{{N + \gamma }}{{1 - \gamma }} > \hat \mu$ must hold.

The iterative process will converge to the equilibrium position, provided that the initial vectors lie in the region of attraction of this equilibrium position. We note that the scheme (27) can be replaced by a similar, more economical from the computational point of view
\[
\left\{ {\begin{array}{*{20}{c}}
{{{\hat x}_n} = \sum_{j = 1}^N {{a_j}{x_{n - j + 1}}},}\\
{{x_{n + 1}} = {{\hat x}_n} - (1 - \gamma ){{\left[ {F'\left({{{\hat x}_n}}\right)} \right]}^ * }F\left({{{\hat x}_n}}\right), n > N}
\end{array}} \right..
\]

\subsection{A generalized method for a simple iteration of the solution of systems of linear equations}

If the system (24) is linear, i.e. $Ax - b = 0$, then the system (26) takes the form
\[
{x_{n + 1}} = \left({I - {A^ * }A}\right){x_n} + {A^*}b,
\]
then, accordingly, the control system (27) becomes
\begin{equation}\label{(28)}
\left\{ {\begin{array}{*{20}{c}}
{ {{\hat x}_n} = \sum_{j = 1}^N {{a_j}{x_{n - j + 1}}}, }\\
{{x_{n + 1}} = \left({I - (1 - \gamma ){A^ * }A}\right){{\hat x}_n} + (1 - \gamma ){A^ * }b, n > N}
\end{array}} \right..
\end{equation}
In the case when the matrix $A$ is symmetric positive definite, the iteration scheme is simplified
\begin{equation}\label{(29)}
\left\{{\begin{array}{*{20}{c}}
{{{\hat x}_n} = \sum_{j = 1}^N {{a_j}{x_{n - j + 1}}}, }\\
{{x_{n + 1}} = \left({I - (1 - \gamma )A}\right){{\hat x}_n} + (1 - \gamma )b, n > N}
\end{array}} \right.
\end{equation}
A similar scheme is also suitable for inversion of matrices
\begin{equation}\label{(30)}
\left\{ {\begin{array}{*{20}{c}}
{ {{\hat X}_n} = \sum_{j = 1}^N {{a_j}{X_{n - j + 1}}}, }\\
{{X_{n + 1}} = \left({I - (1 - \gamma ){A^ * }A}\right){{\hat X}_n} + (1 - \gamma ){A^ * }, n > N}
\end{array}} \right.
\end{equation}
or for a symmetric positive definite matrix $A$
\begin{equation}\label{(31)}
\left\{ {\begin{array}{*{20}{c}}
{ {{\hat X}_n} = \sum_{j = 1}^N {{a_j}{X_{n - j + 1}}}, }\\
{{X_{n + 1}} = \left({I - (1 - \gamma )A}\right){{\hat X}_n} + (1 - \gamma )I, n > N}
\end{array}} \right.
\end{equation}
where ${X_n}$ is a matrix.

Theoretically, the iterative processes (28), (29), (30), (31) converge for any initial values, unlike the usual simple iteration schemes. One advantage of these schemes over other methods of solving linear equations is the absence of division operations in computational processes, which makes it possible to carry out calculations with ill posed matrices.

We note that for $\gamma = 0$, $N = 1$, the generalized simple iteration method coincides with the classical simple iteration method.

\subsection{The generalized Seidel method for solving systems of linear equations}

Suppose that the diagonal elements of $A$ are nonzero. We represent the matrix $A$ in the form
\[
A = L + \hat D + U,
\]
where $\hat D$ is a diagonal matrix, the matrices $L$ and $U$ are lower and upper triangular matrices with zero diagonals.

The classical Seidel method consists of assigning the initial vector ${x_0}$ and sequentially computing the vectors ${x_n}$: $(L + \hat D){x_{n + 1}} = - U{x_n} + b$, and then ${x_{n + 1}} = - {(L + \hat D)^{ - 1}}U{x_n} + {(L + \hat D)^{ - 1}}b$. Of course, this method does not need to build the matrix ${(L + \hat D)^{ - 1}}$. The Seidel method converges if all the eigenvalues of the matrix ${(L + \hat D)^{ - 1}}U$ lie in the central unit disc of the complex plane. This condition is satisfied, for example, if the matrix $A$ is symmetric positive definite.

Let us generalize the Seidel method. We apply to the system ${x_{n + 1}} = - {(L + \hat D)^{ - 1}}U{x_n} + {(L + \hat D)^{ - 1}}b$ the computational scheme (21). We get
\[
\left\{
\begin{array}{*{20}{c}}
\hat{x}_n = \sum_{j = 1}^N {{a_j}{x_{n - j + 1}}},\\
{x_{n + 1}} = \left({\gamma I - (1 - \gamma ){{(L + \hat D)}^{ - 1}}U}\right) {{\hat x}_n} + (1 - \gamma ){{(L + \hat D)}^{ - 1}}b, \quad n > N.
\end{array}
\right.
\]
After simple transformations, the generalized method of P.L. Seidel is reduced to an iterative scheme
\begin{equation}\label{(32)}
\left\{
\begin{array}{*{20}{c}}
\hat x_n = \sum_{j = 1}^N a_j x_{n-j+1},\\
(L + \hat D){x_{n + 1}} = \left({ - U + \gamma A}\right){{\hat x}_n} + (1 - \gamma )b, \quad n > N.
\end{array}
\right.
\end{equation}
For $\gamma = 0$, $N = 1$, the generalized Seidel method coincides with the classical one.

If the matrix $A$ is symmetric positive definite, then it is sufficient to take $N = 1$ in (32).

Let us study the question of the convergence of the iterative scheme (32). Let ${\mu_1}, \ldots, {\mu_m}$ be the eigenvalues of the matrix $ - {(L + \hat D)^{ - 1}}U$. We consider the polynomial (18) for $p({\lambda ^{ - 1}}) = {a_1} {\lambda ^{ - 1}} + \ldots + {a_N}{\lambda ^{ - N}}$, $q({\lambda ^{ - 1}}) = {a_1} + \ldots + {a_N}{\lambda ^{ - N + 1}}$. If all zeros of this polynomial lie in the central unit circle, then the iterative scheme (32) converges. For a suitable choice of $N$, ${a_1}, \ldots, {a_N}$, the scheme (32) converges if the eigenvalues of the matrix $ - {(L + \hat D)^{ - 1}}U$ lie, for example, in the set $M \subseteq \left\{ {\mu \in \mathbb{C}: \left| \mu \right| < 1} \right\} \cup \left\{ {\mu \in \mathbb{R}: \mu \in \left({ - \hat \mu, 1}\right)} \right\}$, $\hat \mu > 1$, or in the set $M \subseteq \left\{ {\mu \in \mathbb{C}: \left| \mu \right| < 1} \right\} \cup \left\{ {\mu \in \mathbb{C}: \left| {\mu + R} \right| < R} \right\}$, $R > {1 \mathord{\left/{\vphantom {1 2}} \right.\kern-\nulldelimiterspace} 2}$.

For the inversion of the matrix $A = L + \hat D + U$, we can apply the scheme
\begin{equation}\label{(33)}
\left\{ {\begin{array}{*{20}{c}}
{ {{\hat X}_n} = \sum_{j = 1}^N {{a_j}{X_{n - j + 1}}}, }\\
(L + \hat D){X_{n + 1}} = \left({ - U + \gamma A}\right){{\hat X}_n} + (1 - \gamma )I, \quad n > N.
\end{array}} \right. 
\end{equation}
We note that the complexity of generalized methods increases quite insignificantly compared with the classical ones, at each iteration it is additionally necessary to perform several addition and multiplication operations.

Similarly, we can generalize other stationary, and even nonstationary, iterative methods for solving systems of algebraic equations

\section{Numerical simulation}

Example 1. Consider a system of nonlinear equations
\begin{equation}
\left\{ {\begin{array}{*{20}{c}}
{{f_i}\left({x, y, z}\right) = 0}\\
{i = 1, 2, 3,}
\end{array}} \right.
\end{equation}
where ${f_1}\left({x, y, z}\right) = - x + {x^3} + {y^2} + 7{z^4} - 1$, ${f_2}\left({x, y, z}\right) = x - y + 2z$, ${f_3}\left({x, y, z}\right) = {(x - y - 8z)^4} - z$. This system was studied in \cite{course}. For its solution we used simple iteration and Newton methods. These methods were used to find the solutions $\left({1, 1, 0}\right)$ and $\left({ - 1, - 1, 0}\right)$, where it was noted that for the convergence, the initial approximation $\left({{x_0}, {y_0}, {z_0}}\right)$ must be close to the solution. This is especially true for ${z_0}$.

To solve system (34) we apply the iterative process (27). We calculate
\[
{\left[ {F'\left({x, y, z}\right)} \right]^ * } = \left({\begin{array}{*{20}{c}}
{ - 1 + 3{x^2}}&1&{4{{(x - y - 8z)}^3}}\\
{2y}&{ - 1}&{ - 4{{(x - y - 8z)}^3}}\\
{28{z^3}}&2&{ - 32{{(x - y - 8z)}^3} - 1}
\end{array}}\right)
\]
and set $N = 3$, $\gamma = 0.91$, $\sigma = 1.4$ in (27). Then ${a_1} \approx 0.46798$, ${a_2} \approx 0.37603$, ${a_3} \approx 0.15600$. For the initial approximations we take three points
\[
\left({{x_0}, {y_0}, {z_0}}\right) = \left({1.55, 0.74, 0.12}\right),
$$
$$ 
\left({{x_0}, {y_0}, {z_0}}\right) = \left({0.84, 0.8, - 0.01}\right),
$$ 
$$ \left({{x_0}, {y_0}, {z_0}}\right) = \left({ - 0.91, - 1.1, - 0.005}\right).
\]
Next, for each of these points, put $\left({{x_i}, {y_i}, {z_i}}\right) = \left({{x_0}, {y_0}, {z_0}}\right)$, $i = 1, 2$. Then the iterative process (27) converges, and for each initial point to a different solution: for the first one to $\left({0.95134, 1.04417, 0.04642}\right)$, for the second one to $\left({1, 1, 0}\right)$ and for the third one to $\left({ - 1, - 1, 0}\right)$. The graphs $\left({n, {x_n}}\right)$, $\left({n, {y_n}}\right)$, $\left({n, {z_n}}\right)$ for different the initial points are shown in Fig. 10 (for the first -- black, for the second -- green, for the third -- blue).

\begin{figure}[h]
\begin{minipage}{0.3\textwidth}
\includegraphics[scale=.5]{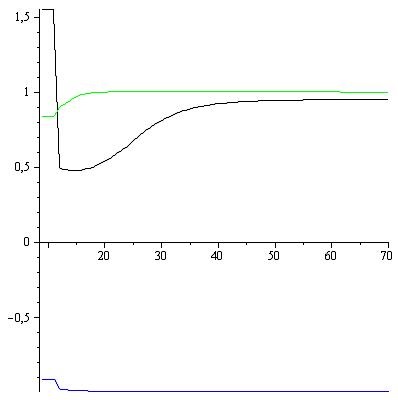}\\
\centering a)
\end{minipage}
\hfil
\begin{minipage}{0.3\textwidth}
\includegraphics[scale=.5]{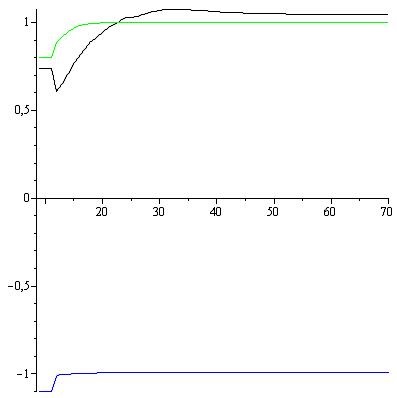}\\
\centering b)
\end{minipage}
\hfil
\begin{minipage}{0.3\textwidth}
\includegraphics[scale=.5]{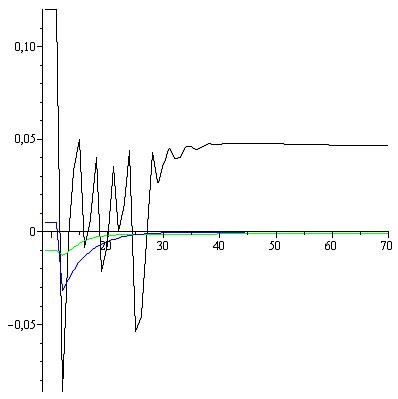}\\
\centering c)
\end{minipage}
\caption{The graphs a) $\left({n, {x_n}}\right)$, b) $\left({n, {y_n}}\right)$, c) $\left({n, {z_n}}\right)$ of the iterative process (27) of the solution of the system (34) for different initial points}
\end{figure}

Fig. 11 shows the graphs of the discrepancy $\left({n, {\varepsilon_n}}\right)$, where
\[
{\varepsilon_n} = \left| {{f_1}({x_n}, {y_n}, {z_n})} \right| + \left| {{f_2}({x_n}, {y_n}, {z_n})} \right| + \left| {{f_3}({x_n}, {y_n}, {z_n})} \right|.
\]

\begin{figure}[h]
\includegraphics[scale=.7]{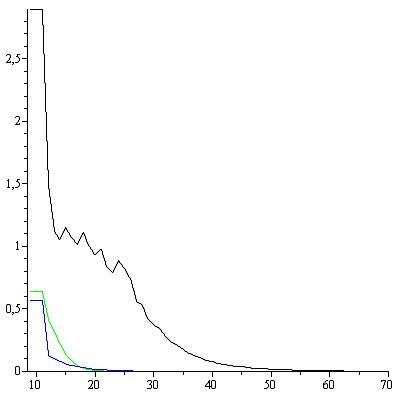}
\caption{Graphs of the residual $\left({n, {\varepsilon_n}}\right)$ of the iterative process (27) of the solution of system (34) for different initial points}
\end{figure}

We also give the values of the first iterations and the discrepancy graph for the iterative process (26) and the initial vector $\left({{x_0}, {y_0}, {z_0}}\right) = \left({1.00001, 0.99999, 0}\right)$: $\left({{x_7}, {y_7}, {z_7}}\right) = \left({1.086, 0.910, 0.246}\right)$, $\left({{x_8}, {y_8}, {z_8}}\right) = \left({234.865, - 233.087, - 1867.571}\right)$.

\begin{figure}[h]
\includegraphics[scale=.7]{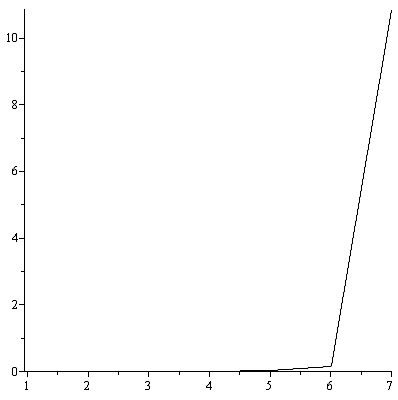}
\caption{The discrepancy graph of the $\left({n, {\varepsilon_n}}\right)$ iterative process (26) of the system solution (34)}
\end{figure}

Compared with the simple iteration method and the Newton method, the proposed method turned out to be more efficient, allowing us to find one more solution, and the basin of attraction of the equilibrium turns out to be much larger.\\

Example 2. Consider the matrix
\begin{equation}
A = \left({\begin{array}{*{20}{c}}
1&2&3\\
2&{ - 2}&{ - 10}\\
3&{ - 10}&1
\end{array}}\right)
\end{equation}
and let us apply the iterative process (33) for its inversion. Since the eigenvalues of the matrix $A$ are $\left\{ { - 11.58, 1.85, 9.73} \right\}$, the method of simple iterations of the inversion of this matrix will diverge. We expand the matrix $A$ as
\[
A = L + \hat D + U = \left({\begin{array}{*{20}{c}}
0&0&0\\
2&0&0\\
3&{ - 10}&0
\end{array}}\right) + \left({\begin{array}{*{20}{c}}
1&0&0\\
0&{ - 2}&0\\
0&0&1
\end{array}}\right) + \left({\begin{array}{*{20}{c}}
0&2&3\\
0&0&{ - 10}\\
0&0&0
\end{array}}\right),
\]
and find the eigenvalues of the matrix $ - {\left({L + \hat D}\right)^{ - 1}} U$: $\left\{ {0, - 0.41, - 72.59} \right\}$. Since these eigenvalues do not lie in the central unit circle, the Seidel method is not applicable for inversion of the matrix (35). But these eigenvalues are less than unity, therefore, we apply the generalized Seidel method.

In the formula (33) we take $N = 7$, $\gamma = 0.743$, $\sigma = 1.8$. Then ${a_1} \approx 0.14722$, ${a_2} \approx 0.21348$, ${a_3} \approx 0.22286$, ${a_4} \approx 0.19052$, ${a_5} \approx 0.13372$, ${a_6} \approx 0.07116$.

As initial approximations, we take matrices: $X[1]$ is a unit matrix, $X[2], \ldots, X[7]$ is zero. Then
\[
X[250] = \left({\begin{array}{*{20}{c}}
{0.490}&{0.154}&{0.067}\\
{0.154}&{0.038}&{ - 0.077}\\
{0.067}&{ - 0.077}&{0.029}
\end{array}}\right).
\]
Denote by ${\varepsilon_n} = {\left\| {X[n] A - I} \right\|_1}$, where the norm ${\varepsilon_n} = {\left\| \bullet \right\|_1}$ is defined as the sum of the absolute values of the matrix components. We calculate $\varepsilon_{250}\approx 3\cdot 10^{-9}$. To visualize the convergence of the matrix inversion process, we plot the discrepancy graph

\begin{figure}[h]
\includegraphics[scale=0.7]{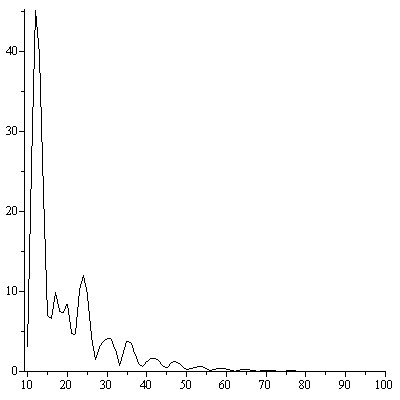}
\caption{The discrepancy graph $\left({n, {\varepsilon_n}}\right)$ of the iterative process (33) of the inverse of the matrix (35)}
\end{figure}

Because of poor initial approximation, the discrepancy at the first few steps increased sharply, however, after ten steps this discrepancy began to decrease rapidly. This confirms the practical effectiveness of the proposed iterative scheme. We note that, both with increasing $\gamma$, and with decreasing $N$, the rate of convergence will decrease.

For comparison, we give numerical calculations using Morgul's scheme, i.e. In the formula (33) we take $N = 1$. For this case, the best value for $\gamma$ will be $0.974$. The required accuracy is achieved at 800 step. We give the graphs of the discrepancy of the previous scheme and Morgul's scheme for the first 80 iterations.

\begin{figure}[h]
\includegraphics[scale=0.7]{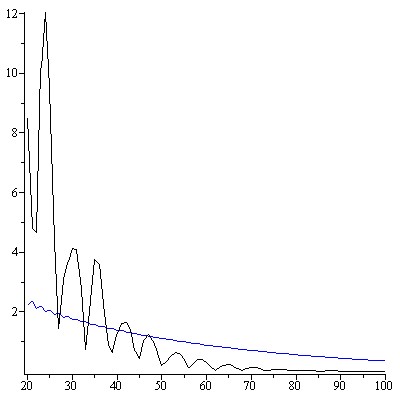}
\caption{Graphs of the discrepancy $\left({n, {\varepsilon_n}}\right)$ of the iterative process (33) of the inverse of the matrix (35) for $N = 7$ (black) and $N = 1$ (blue)}
\end{figure}

It can be seen that the discrepancy of Morgul's method decreases monotonically, but is much slower than the discrepancy of the generalized semilinear control.

\section{Conclusion}

In the article the problem of stabilization of unstable and a priori unknown periodic orbits of nonlinear systems with discrete time is considered. A new approach to constructing delayed feedback, which solves the stabilization problem, is proposed. The feedback is represented as a convex combination of nonlinear feedback and semilinear feedback introduced by O. Morgul. This preserves the advantages of both types of feedback.

The methods of geometric complex analysis were used to construct the nonlinear feedback gain factors and to obtain the conditions for the applicability of such control. These methods are used to analyze the possibility of using Morgul's scheme. The necessary and sufficient conditions for stabilization in the form of a geometric criterion for local asymptotic stability are obtained. Morgul's method was also transferred from the scalar case to the vector one.

It is important to note that the characteristic polynomials for periodic orbits in the nonlinear and semilinear cases have a very simple structure, although, naturally, different. It was this circumstance that stimulated the integration of the two approaches mentioned above. The resulting characteristic polynomial also has a rather simple structure and contains, as special cases, polynomials of nonlinear and semilinear control schemes.

The geometric criterion of stability in the nonlinear and semilinear cases consists of the analysis of images of the central unit circle under a special polynomial mapping. In a combined nonlinear-semilinear control method, instead of polynomial mappings, one has to study rational mappings. In this paper, we give a solution to the construction of quasioptimal fractional-rational maps for the case $T = 1$, that is, to stabilize the equilibrium positions. An additional introduction to the control of semilinear feedback allows us to significantly reduce the length of the  used prehistory in the delayed feedback and to increase the rate of convergence of the perturbed solutions to the periodic ones.

As an application of the proposed stabilization scheme, a possible computational algorithm for finding solutions of systems of algebraic equations is presented, based on the modification of known iterative schemes. In these schemes, the values of the variables computed in the previous steps are used. At the same time, the complexity of the new iterative schemes practically does not increase.

The above results of numerical solutions of systems of linear and nonlinear equations confirm our solution as an improvement of previous work and the effectiveness of the proposed equilibrium stabilization schemes.

\section{acknowledgment}
The authors are deeply grateful to Alexei Solyanik and Emil Iacob for their valuable comments and help in preparation of manuscript.

\bigskip

D. Dmitrishin, I. Skrynnik and E. Franzheva,  Odessa National Polytechnic University, 1 Shevchenko Ave, Odessa 65044, Ukraine.
e-mail: dmitrishin@opu.ua \\

A. Stokolos, Georgia Southern University, Statesboro, GA 30458, USA. e-mail: astokolos@georgiasouthern.edu
\end{document}